\documentclass[english]{article}
\pdfoutput=1
\usepackage[T1]{fontenc}
\usepackage[latin9]{inputenc}
\usepackage[a4paper]{geometry}
\usepackage{babel}
\usepackage{amsmath}
\usepackage{amsthm}
\usepackage[unicode=true,pdfusetitle,
 bookmarks=true,bookmarksnumbered=false,bookmarksopen=false,
 breaklinks=false,pdfborder={0 0 1},backref=false,colorlinks=false]
 {hyperref}

\makeatletter
\numberwithin{equation}{section}
\numberwithin{figure}{section}
\newcommand{\lyxaddress}[1]{
	\par {\raggedright #1
	\vspace{1.4em}
	\noindent\par}
}
\theoremstyle{plain}
\newtheorem{thm}{\protect\theoremname}
\theoremstyle{plain}
\newtheorem{lem}{\protect\lemmaname}
\theoremstyle{remark}
\newtheorem{claim}{\protect\claimname}
\theoremstyle{plain}
\newtheorem{cor}{\protect\corollaryname}

\usepackage{amsthm}
\usepackage{amsfonts}
\usepackage{graphicx}
\usepackage{float}
\usepackage{bbm}

\let\originalleft\left
\let\originalright\right
\renewcommand{\left}{\mathopen{}\mathclose\bgroup\originalleft}
\renewcommand{\right}{\aftergroup\egroup\originalright}

\usepackage{amsmath,amssymb}

\DeclareMathOperator{\E}{\mathbb{E}}

\DeclareMathOperator{\CIS}{\textsc{CIS}}

\DeclareMathOperator{\CC}{\textsc{C}}

\date{}

\makeatother

\providecommand{\claimname}{Claim}
\providecommand{\corollaryname}{Corollary}
\providecommand{\lemmaname}{Lemma}
\providecommand{\theoremname}{Theorem}

\begin{document}

\title{Several Separations Based on a Partial Boolean Function\thanks{Supported by the project ``Quantum algorithms: from complexity theory
to experiment'' funded under ERDF programme 1.1.1.5.}}

\author{Kaspars Balodis$^{1}$}
\maketitle

\lyxaddress{$^{1}$ Center for Quantum Computer Science, Faculty of Computing,
University of Latvia}
\begin{abstract}
We show a partial Boolean function $f$ together with an input $x\in f^{-1}\left(*\right)$
such that both $\CC_{\bar{0}}\left(f,x\right)$ and $\CC_{\bar{1}}\left(f,x\right)$
are at least $\CC\left(f\right)^{2-o\left(1\right)}$. Due to recent
results by Ben-David, G\"{o}\"{o}s, Jain, and Kothari, this result
implies several other separations in query and communication complexity.
For example, it gives a function $f$ with $\CC(f)=\Omega(\deg^{2-o\left(1\right)}(f))$
where $\CC$ and $\deg$ denote certificate complexity and polynomial
degree of $f$. (This is the first improvement over a separation between
$\CC(f)$ and $\deg(f)$ by Kushilevitz and Nisan in 1995.) Other
implications of this result are an improved separation between sensitivity
and polynomial degree, a near-optimal lower bound on conondeterministic
communication complexity for Clique vs. Independent Set problem and
a near-optimal lower bound on complexity of Alon--Saks--Seymour
problem in graph theory.
\end{abstract}

\section{The puzzle}

Recently Ben-David, G\"{o}\"{o}s, Jain, and Kothari published a paper
\cite{ben2021unambiguous} demonstrating that several separation problems
can be reformulated (disguised) as one of three equivalent puzzles,
hinting that such formulations may be more seductive for tricking
more people into trying to solve them. We report that they have indeed
succeeded and show an optimal solution to one of the puzzles.

We use the following formulation from \cite{ben2021unambiguous}.
Consider a partial Boolean function $f:\left\{ 0,1\right\} ^{n}\rightarrow\left\{ 0,1,*\right\} $
where some of the inputs are \emph{undefined}, $f\left(x\right)=*$.
Let $\Sigma\subseteq\left\{ 0,1,*\right\} $ be a subset of output
symbols. Denote by $0$, $1$, $\bar{0}$, $\bar{1}$ the output sets
$\left\{ 0\right\} $, $\left\{ 1\right\} $, $\left\{ 1,*\right\} $,
$\left\{ 0,*\right\} $. A partial input $\rho\in\left\{ 0,1,*\right\} ^{n}$
is a $\Sigma$-\emph{certificate} for $x\in\left\{ 0,1\right\} ^{n}$
if $\rho$ is consistent with $x$ (i.e., for each entry in $x$,
the corresponding entry of $\rho$ contains the same symbol or $*$)
and for every input $x'$ consistent with $\rho$ we have $f\left(x'\right)\in\Sigma$.
The \emph{size} of $\rho$, denoted $\left|\rho\right|$, is the number
of its non-$*$ entries. The $\Sigma$-\emph{certificate complexity}
of $x$, denoted $\CC_{\Sigma}\left(f,x\right)$, is the least size
of a $\Sigma$-certificate for $x$. The $\Sigma$-\emph{certificate
complexity} of $f$, denoted $\CC_{\Sigma}\left(f\right)$ is the
maximum of $\CC_{\Sigma}\left(f,x\right)$ over all $x\in f^{-1}\left(\Sigma\right)$.
Finally, we define the\emph{ certificate complexity} $\CC\left(f\right)$
as $\max\left\{ \CC_{0}\left(f\right),\CC_{1}\left(f\right)\right\} $.

\textbf{Puzzle}. For $\alpha>1$, does there exist a partial function
$f$ together with an $x\in f^{-1}\left(*\right)$ such that both
$\CC_{\bar{0}}\left(f,x\right)$ and $\CC_{\bar{1}}\left(f,x\right)$
are at least $\CC\left(f\right)^{\alpha-o\left(1\right)}$?

Abusing the terminology, instead of a single Boolean function actually
an infinite sequence of functions $f_{n}$ satisfying $\CC\left(f_{n}\right)\rightarrow\infty$
as $n\rightarrow\infty$ is meant. It is known that a solution with
$\alpha=2$ would be optimal. In \cite{ben2021unambiguous} a simple
function with $\alpha=1.5$ inspired by the board game Hex is constructed.
It is also conjectured that the puzzles are soluble with exponent
$2$. Indeed, this is the case and we demonstrate a function achieving
the optimal $\alpha=2$.

\section{The solution}

Our contribution is as follows.
\begin{thm}
\label{thm:function}There exists a monotone partial Boolean function
$f$ and an input $x\in f^{-1}\left(*\right)$ such that both $\CC_{\bar{0}}\left(f,x\right)$
and $\CC_{\bar{1}}\left(f,x\right)$ are at least $\CC\left(f\right)^{2-o\left(1\right)}$.
\end{thm}
We denote by $\left[n\right]$ the set $\left\{ 1,2,\dots,n\right\} $.

Let $\left(r_{k}:\left[n\right]\times\left[n\right]\rightarrow\left[n\right]\right)_{k\in\left[\ell\right]}$
be a collection of $\ell$ independent random functions where the
output is chosen uniformly from $\left[n\right]$.
\begin{lem}
\label{lemma:rand}Let $\ell>4$. Consider an arbitrary $S\subseteq\left[n\right]$
with $\left|S\right|=m\leq n^{\frac{\ell+1}{\ell+2}}$. With probability
$1-o\left(1\right)$, 
\[
\left|\left\{ \left(i,j\right)\mid i,j\in S\wedge\forall k\in\left[\ell\right]\,r_{k}\left(i,j\right)\in S\right\} \right|\leq\ell\cdot n.
\]
\end{lem}
\begin{proof}
First, let us consider a random $S\subseteq\left[n\right]$. Define
a random variable
\[
Z_{i,j}^{S}=\begin{cases}
1, & \text{if \ensuremath{\forall k\in\left[\ell\right]\,r_{k}\left(i,j\right)\in S}}\\
0, & \text{otherwise}
\end{cases}.
\]

$\E\left[Z_{i,j}^{S}\right]=\Pr\left[Z_{i,j}^{S}=1\right]=\left(\frac{m}{n}\right)^{\ell}$.

Let $Z^{S}=\sum_{i,j\in S}Z_{i,j}^{S}.$

\begin{align*}
\E\left[Z^{S}\right] & =\sum_{i,j\in S}\E\left[Z_{i,j}^{S}\right]\\
 & =m^{2}\cdot\left(\frac{m}{n}\right)^{\ell}\\
 & \leq\left(n^{\frac{\ell+1}{\ell+2}}\right)^{2}\cdot\left(\frac{n^{\frac{\ell+1}{\ell+2}}}{n}\right)^{\ell}\\
 & =n^{\frac{2\ell+2}{\ell+2}}\cdot n^{\frac{-\ell}{\ell+2}}\\
 & =n.
\end{align*}

Therefore, $\mu_{u}=n$ is an upper bound on $\E\left[Z^{S}\right]$.

By the Chernoff inequality (see e.g. \cite{chung2006concentration})
$\Pr\left[Z^{S}\geq\left(1+\varepsilon\right)\mu_{u}\right]\leq\exp\left(-\frac{\varepsilon^{2}}{2+\varepsilon}\mu_{u}\right)$.

Therefore, 
\begin{align*}
\Pr\left[Z^{S}\geq\ell\cdot n\right] & =\Pr\left[Z^{S}\geq\left(1+\left(\ell-1\right)\right)\mu_{u}\right]\\
 & \leq\exp\left(-\frac{\left(\ell-1\right)^{2}}{2+\ell-1}n\right)\\
 & =\exp\left(-\frac{\left(\ell-1\right)^{2}-4+4}{\ell+1}n\right)\\
 & \leq\exp\left(-\frac{\left(\ell-1\right)^{2}-4}{\ell+1}n\right)\\
 & =\exp\left(-\frac{\left(\ell-1-2\right)\left(\ell-1+2\right)}{\ell+1}n\right)\\
 & =\exp\left(-\left(\ell-3\right)n\right).
\end{align*}

Now, let us calculate the probability that there exists such $S$
that violates the inequality in the Lemma statement.

\[
\Pr\left[\exists S:Z^{S}\geq\ell\cdot n\right]\leq\frac{2^{n}}{e^{\left(\ell-3\right)n}}\leq\frac{e^{n}}{e^{\left(\ell-3\right)n}}=\frac{1}{e^{\left(\ell-4\right)n}}=o\left(1\right).
\]
\end{proof}
Consider an input consisting of $2n^{2}$ variables $x_{i,j,b}\in\left\{ 0,1\right\} $
with $i,j\in\left[n\right]$ and $b\in\left[2\right]$. The input
is interpreted as an $n\times n$ matrix containing pairs of Boolean
values as entries. The variable pair $\left(x_{i,j,1},x_{i,j,2}\right)$
is the entry in the $i$-th row and $j$-th column. We refer to the
$i$-th row by $x_{i}$.

We call two entries $\left(a_{1},a_{2}\right)$ and $\left(b_{1},b_{2}\right)$
\emph{matching }if $\left(a_{1}\wedge b_{1}\right)\vee\left(a_{2}\wedge b_{2}\right)$.
We call two distinct rows $x_{i_{1}}$ and $x_{i_{2}}$ \emph{matching
}if in each column they have matching entries.

We call a row \emph{bad }if it contains an entry $\left(0,0\right)$.

For every pair $\left(i_{1},i_{2}\right)$, we call the rows $x_{r_{1}\left(i_{1},i_{2}\right)},\dots,x_{r_{\ell\left(i_{1},i_{2}\right)}}$
\emph{associated }with the rows $x_{i_{1}},x_{i_{2}}$.

Define $f\left(x\right)=1$ if there exist two matching rows $x_{i_{1}},x_{i_{2}}$,
and none of the associated rows $x_{r_{1}\left(i_{1},i_{2}\right)},\dots,x_{r_{\ell}\left(i_{1},i_{2}\right)}$
are bad.

Notice that a bad row dismisses its chance to be matching with any
other row, as well as spoils every pair for which it is associated.

Define $f\left(x\right)=0$ if there exists a certificate on at most
$\left(2\ell+2\right)n$ variables which certifies that $f\left(x\right)\neq1$.

Otherwise, define $f\left(x\right)=*$.

More formally,

\[
f\left(x\right)=\begin{cases}
1, & \exists i_{1},i_{2}\in\left[n\right]:\,\left(i_{1}\neq i_{2}\right)\\
 & \wedge\left(\forall j\in\left[n\right]\,\exists b\in\left[2\right]\,\left(x_{i_{1},j,b}\wedge x_{i_{2},j,b}\right)\right)\\
 & \wedge\left(\forall k\in\left[\ell\right]\,\forall j\in\left[n\right]\,\exists b\in\left[2\right]\,x_{r_{k}\left(i_{1},i_{2}\right),j,b}\right)\\
0, & \text{a certificate with \ensuremath{\leq\left(2\ell+2\right)n} variables exists certifying that }f\left(x\right)\neq1\\
*, & \text{otherwise}
\end{cases}.
\]

$\CC_{1}\left(f\right)\leq2n\cdot\left(\ell+2\right)$ because the
two matching rows $x_{i_{1}},x_{i_{2}}$ together with the associated
rows $x_{r_{1}\left(i_{1},i_{2}\right)},\dots,x_{r_{\ell}\left(i_{1},i_{2}\right)}$
certify that $f\left(x\right)=1$.

$\CC_{0}\left(f\right)\leq\left(2\ell+2\right)n$ by definition.

Consider the input $z$ in which $z_{i,i,1}=1,z_{i,i,2}=0$, and $z_{i,j,1}=0,z_{i,j,2}=1$
for $i\neq j$. I.e., the diagonal entries are $\left(1,0\right)$,
and all other entries are $\left(0,1\right)$. Clearly, $f\left(z\right)\neq1$
as every pair of rows are non-matching, due to the diagonal entries.

$\CC_{\bar{1}}\left(f,z\right)\geq\frac{n\left(n-1\right)}{2}$ because
there are no bad rows and every pair of rows $z_{i_{1}},z_{i_{2}}$
could be made matching by setting $z_{i_{1},i_{2},1}=z_{i_{2},i_{1},1}=1$,
therefore any certificate certifying that $f\left(z\right)\neq1$
should contain at least one of $z_{i_{1},i_{2},1}$ and $z_{i_{2},i_{1},1}$
for every $i_{1},i_{2}\in\left[n\right]\left(i_{1}\neq i_{2}\right)$.
\begin{claim}
$\CC_{\bar{0}}\left(f,z\right)>n^{\frac{2\ell+3}{\ell+2}}$.
\end{claim}
\begin{proof}
Let $\rho$ be a partial input consistent with $z$ that has size
$n^{\frac{2\ell+3}{\ell+2}}$. We will construct a $0$-certificate
$\sigma$ consistent with $\rho$ with size $\leq\left(2\ell+2\right)n$.
By averaging, there exists a column in which $\rho$ has read (one
or both variables) from $m\leq n^{\frac{\ell+1}{\ell+2}}$ entries.
In this column in the unseen entries we write $\left(0,0\right)$
in $\sigma$, therefore making the corresponding rows bad, hence unfit
for being matching. Here we have used no more than $2n$ variables
in $\sigma$.

Now we have to spoil the remaining $m\left(m-1\right)$ possible pairs.
Notice that, by association, most of these pairs are already spoiled.
By Lemma \ref{lemma:rand} there are at most $\ell\cdot n$ possibly
matching unspoiled pairs of rows (the rest are spoiled by having at
least one bad row associated with them). We spoil each of them by
exposing in $\sigma$ two zeros that make these rows non-matching.
With this, the $0$-certificate $\sigma$ is complete -- every pair
of rows are shown to be non-matching or having an associated bad row.
In this step we have used at most $2\ell n$ variables in $\sigma$,
and at most $\left(2\ell+2\right)n$ in total.

Therefore, we have constructed a $0$-certificate that is consistent
with $\rho$ and has size $\leq\left(2\ell+2\right)n$. Therefore,
$\rho$ cannot be a $\bar{0}$-certificate, i.e., any $\bar{0}$-certificate
must have size $>n^{\frac{2\ell+3}{\ell+2}}$.
\end{proof}
By setting $\ell=\log n$, we have 

\begin{align*}
\CC_{1}\left(f\right) & =\tilde{\Omega}\left(n\right)\\
\CC_{0}\left(f\right) & =\tilde{\Omega}\left(n\right)\\
\CC_{\bar{1}}\left(f,z\right) & =n^{2-o\left(1\right)}\\
\CC_{\bar{0}}\left(f,z\right) & =n^{2-o\left(1\right)}
\end{align*}

and Theorem \ref{thm:function} follows.

Notice that $f$ is monotone, i.e., flipping any bit in an input $z$
from $0$ to $1$ can only change $f\left(z\right)$ from $0$ to
$*$ or $1$, or from $*$ to 1. This is no coincidence, because,
in fact, $f$ was derived from a somewhat more complex function by
a transformation inspired by \cite[Remark 15]{ben2021unambiguous}
which transforms a function into a monotone one.

\section{The implications}

In this section we list the main bounds and separations arising from
our result. All of them are noted in \cite{ben2021unambiguous}. As
they cover a wide range of concepts and contain no new contributions
from our side, we restrict ourselves to only listing them and indeed
even do not define all the terminology used for stating them, but
refer the reader to \cite{ben2021unambiguous} and other mentioned
sources instead.

The following two corollaries follow from the other two formulations
of equivalent puzzles in \cite{ben2021unambiguous}.
\begin{cor}
There exists a Boolean function f with $\CC_{0}\left(f\right)\geq\textsc{UC}_{1}\left(f\right)^{2-o\left(1\right)}$.
\end{cor}
\begin{cor}
There exists an intersecting hypergraph $G=\left(V,E\right)$ together
with a colouring $c:V\rightarrow\left\{ 0,1\right\} $ such that every
$c$-monochromatic hitting set has size at least $r\left(G\right)^{2-o\left(1\right)}$.
\end{cor}
The next corollary follows from \cite{goos2015lower} and gives a
near-optimal lower bound for the complexity of the \emph{Clique vs.
Independent Set} problem by Yannakakis \cite{yannakakis1991expressing}.
\begin{cor}
There exists a graph $G$ such that the $\CIS_{G}$ requires $\Omega\left(\log^{2-o\left(1\right)}n\right)$
bits of conondeterministic communication.
\end{cor}
Equivalently (see, e.g., \cite{bousquet2014clique}), the same gap
applies to the graph-theoretic Alon--Saks--Seymour problem.
\begin{cor}
There exists a graph G such that $\chi\left(G\right)\geq\exp\left(\Omega\left(\log^{2-o\left(1\right)}\textrm{bp}\left(H\right)\right)\right)$.
\end{cor}
The next two separations follow from the cheat sheet constructions
\cite{ben2016low,aaronson2016separations}. They improve the power-$2.5$
separation due to \cite{ben2021unambiguous} and the power-$1.63$
separation due to Nisan, Kushilevitz, and Widgerson \cite{nisan1995rank},
respectively.
\begin{cor}
There exists a Boolean function $f$ with $\CC\left(f\right)\geq\Omega\left(s\left(f\right)^{3-o\left(1\right)}\right)$.
\end{cor}
\begin{cor}
There exists a Boolean function $f$ with $\CC\left(f\right)\geq\Omega\left(\deg\left(f\right)^{2-o\left(1\right)}\right)$.
\end{cor}
\bibliographystyle{plain}
\bibliography{cert-ref}

\end{document}